\documentclass[runningheads]{llncs}
\usepackage{graphicx}
\usepackage[font=small,labelfont=bf]{caption}
\usepackage{subcaption}
\usepackage[clock]{ifsym}

\usepackage{comment}
\usepackage{amsmath,amssymb,bm,esint}
\usepackage{graphicx}
\usepackage{algorithmic}
\usepackage{algorithm}
\usepackage{setspace}
\usepackage[dvipsnames]{xcolor}
\usepackage{rotating}
\usepackage{blkarray}
\usepackage{multirow}
\usepackage{tabularx}
\usepackage{tablefootnote}
\usepackage{booktabs}
\usepackage{array}
\usepackage{listings}

\usepackage{tikz}
\usetikzlibrary{shapes.geometric, arrows}
\tikzstyle{startstop} = [rectangle, rounded corners, minimum width=3cm, minimum height=1cm,text centered, draw=black, fill=red!30]

\tikzstyle{io} = [rectangle, minimum width=4.5cm, minimum height=1.8cm, text centered,  text width=4cm, draw=black, fill=white!30]

\tikzstyle{process} = [rectangle, minimum width=3cm, text width=8cm, minimum height=1cm, text centered, draw=black, fill=orange!30]

\tikzstyle{decision} = [diamond, minimum width=3cm, minimum height=1cm, text centered, draw=black, fill=green!30]

\tikzstyle{arrow} = [thick,->,>=stealth, text width = 110]

\usepackage[colorlinks=true,linkcolor=red,          
    citecolor=blue,        
    filecolor=magenta,      
    urlcolor=blue]{hyperref}

\newcommand{\E}{\mathbb{E}}
\newcommand{\var}{\mathbb{V}\mathrm{ar}} 

\newcommand{\pr}{\mathbb{P}}

\newcommand{\real}{{\mathbb R}}
\newcommand{\nat}{{\mathbb N}}

\newcommand{\X}{{\mathbf X}}

\newcommand{\Y}{{\mathbf Y}}
\newcommand{\D}{{\mathbf D}}

\newcommand{\Z}{{\mathbf Z}}

\newcommand{\F}{{\mathbf F}}

\newcommand{\Ncal}{\mathcal{N}}
\newcommand{\Pcal}{\mathcal{P}}

\newcommand{\miroslav}[1]{{\color{Brown}#1}}

\begin{document}
%
\title{Moment-based Invariants for Probabilistic Loops with Non-polynomial Assignments\thanks{Supported by the Vienna Science and Technology Fund (WWTF ICT19-018), the TU Wien Doctoral College (SecInt), the FWF research
projects LogiCS W1255-N23 and P 30690-N35, and the ERC Consolidator Grant ARTIST 101002685.}}
%
 \titlerunning{Moment-based Invariants for PPs with Non-polynomial Updates}
%
\author{Andrey Kofnov\inst{1} 
\and {Marcel Moosbrugger} \inst{2}
\and {Miroslav Stankovi{\v{c}}} \inst{2}
\and \\ {Ezio Bartocci} \inst{2}
\and {Efstathia Bura} \inst{1}
}
\authorrunning{A. Kofnov et al.}
%
\institute{Applied Statistics, Faculty of Mathematics and Geoinformation, TU Wien \\
\and Faculty of Informatics, TU Wien}
\maketitle              
\begin{abstract}

We present a method to automatically approximate moment-based invariants of probabilistic programs with non-polynomial updates of continuous  state variables to accommodate more complex dynamics. Our approach leverages polynomial chaos expansion to approximate  non-linear functional updates as sums of orthogonal polynomials. We exploit this result to automatically estimate  state-variable moments of all orders in Prob-solvable loops with non-polynomial updates. 
We showcase the accuracy of our estimation approach in several examples, such as the turning vehicle model and the Taylor rule in monetary policy.

\keywords{Probabilistic programs  \and Prob-solvable loops \and Polynomial Chaos Expansion \and Non-linear updates}
\end{abstract}
\section{Introduction}

Probabilistic programs (PPs) are becoming widely employed in many areas including AI applications, security/privacy protocols or modeling stochastic dynamical systems.
The study of the properties of these processes requires knowledge of their distribution; that is, the distribution(s) of the random variable(s) generated by executing the probabilistic program.

The characterization of many distributions can be accomplished via their moments.  In~\cite{Bartoccietal2019} the authors introduced a class of probabilistic programs, \textit{Prob-solvable loops}, for which moment-based invariants over the state variables of the programs are automatically computed as a closed-form expression.  A Prob-solvable loop consists of an initialization section and a non-nested loop where the variables can be updated by drawing from  distributions determined by their moments (e.g., Bernoulli, Normal) and using polynomial arithmetic. 
However, modeling complex dynamics often requires the use of non-polynomial updates, such as in the \emph{turning vehicle} example in Fig.~\ref{fig:turningexample}.  An open research question 
is how to leverage the class of \textit{Prob-solvable loops} to estimate moment-based invariants as closed-form expressions for probabilistic loops with updates governed by non-polynomial non-linear functions.


\begin{figure}[!t]
\centering
  \includegraphics[scale=.46]{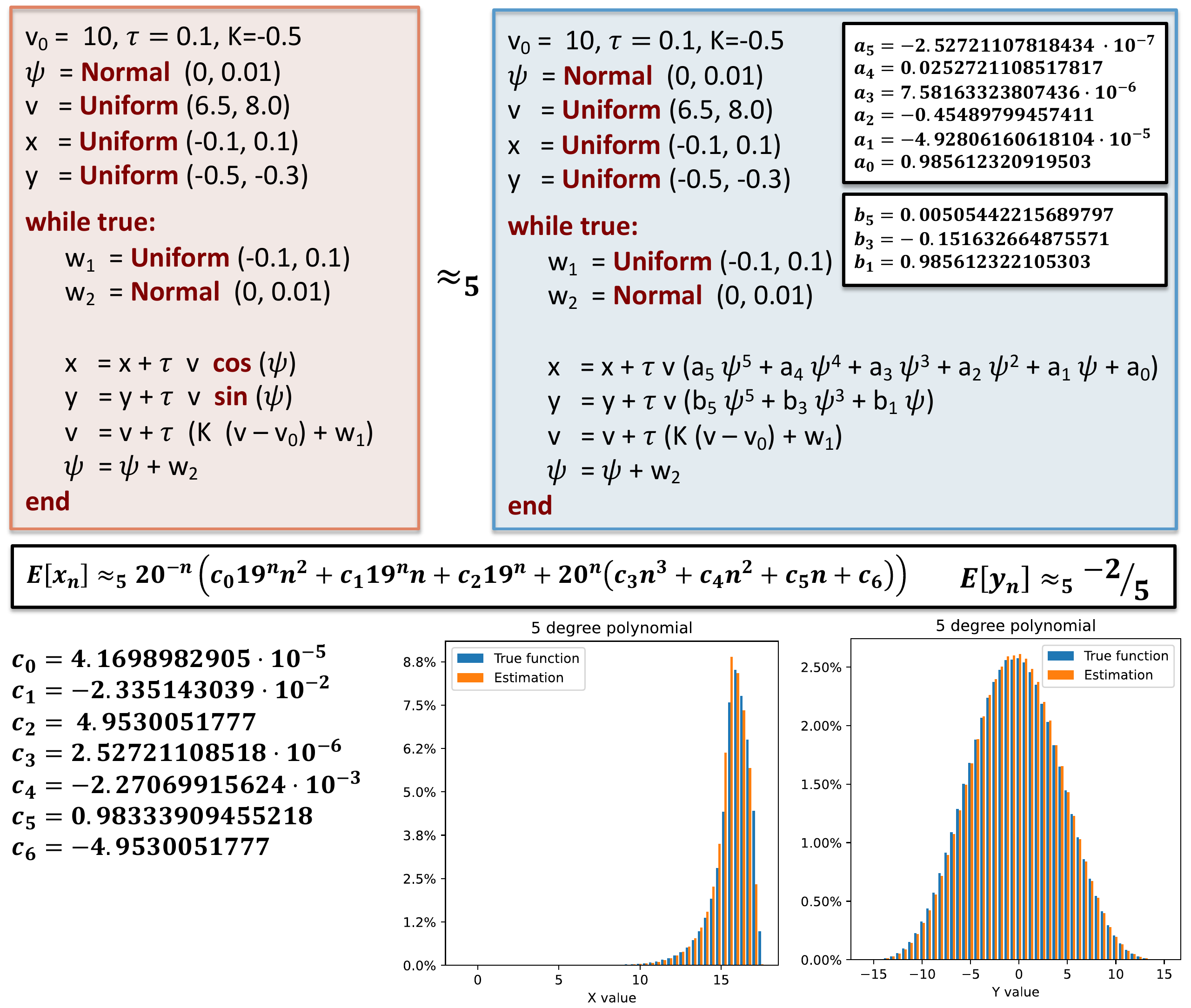}
  \caption{On the top left a probabilistic loop modeling the behaviour of a turning vehicle~\cite{Srirametal2020} using non-polynomial (cos, sin) updates in the loop body. On top right a Prob-Solvable loop obtained by approximating the cos, sin functions using polynomial chaos expansion (up to 5th degree).  In the middle the expected position $(x,y)$
  computed automatically from the Prob-Solvable loop as a closed-form expression in the number of the loop iterations $n$. In the bottom center 
  and right
  the comparison among the true and the estimated distribution for a fixed iteration (we execute the
  loop for $n=20$ iterations and $8 \cdot 10^5$ repetitions).} 
  \label{fig:turningexample}
\end{figure}

At the heart of our approach is the  decomposition of a random function into a linear combination of basis functions that are orthogonal polynomials.   
We accommodate non-polynomial updates of program variables to allow for more complicated dynamics. By expressing the non-linear functionals of the updates as sums of orthogonal polynomials, we can apply the approach in  \cite{Bartoccietal2019,BartocciKovacsStankovic2021} to automatically estimate the moments of all the program random variables. Our approach  is within the framework of general polynomial chaos expansion (gPCE) \cite{XiuKarniadakis2002a}. 
As such, it converges to the truth with guaranteed recovery of the moments of random variables with complex distributions.  We focus on state variables with continuous distributions with updates that are square-integrable functionals and use general polynomial chaos expansion to represent them.  In Fig.~\ref{fig:turningexample} we illustrate our approach via the turning vehicle example, where we estimate the expected position of a vehicle. In this example, we approximate the original cosine and sine functions with 5$th$ degree polynomials and obtain a Prob-solvable loop.  This enables the automatic
computation of the moments in closed-form  at each loop iteration ($n$) using the approach proposed in~\cite{Bartoccietal2019}.

\paragraph{Related Work.} \cite{Srirametal2020} proposed polynomial forms to represent distributions of state variables.
Polynomial forms are generalizations of affine forms~\cite{Bouissouetal2016} and use the Taylor series expansion to represent functions of random variables generated in a PP. 
Functions can be only approximated in a small interval around a fixed point, otherwise high order derivatives are required to guarantee sufficient accuracy of the approximation. 
As a consequence, functions with unbounded support cannot be handled with this approach.

So-called Taylor models have been proposed in \cite{Makino2003,Neher2007,Chenetal12} for reachability analysis of (non-probabilistic) non-linear dynamical systems.
Taylor models combine polynomials and error intervals to capture the set of reachable states after some fixed time horizon.
Application of Taylor series expansions for generalized functions of probabilistic distributions can also be found in~\cite{Stankovic1996,Triebel2001}.

\cite{Jasouretal2021} introduced trigonometric and mixed-trigonometric-polynomial moments to obtain an exact description of the moments of uncertain states for nonlinear autonomous and robotic systems over the planning horizon. This approach can only handle systems encoded in PPs, where all nonlinear transformations take standard, trigonometric, or mixed-trigonometric polynomial forms.

Polynomial chaos expansion based methods have been extensively used for uncertainty quantification in engineering problems of solid and fluid mechanics (e.g. \cite{GhanemSpanos1991,Fooetal2007,Houetal2006}), computational fluid dynamics (e.g., \cite{Knio2006}), flow through porous media \cite{GhanemDham1998,Ghanem1998}, thermal problems
\cite{HienKleiber1997}, analysis of turbulent velocity
fields \cite{Chorin1974,MeechamJeng1968}, differential equations (e.g., \cite{WanKarniadakis2005,XiuKarniadakis2002a}), and, more recently, geosciences and meteorology (e.g., \cite{Formaggiaetal2013,Giraldietal2017,Denamieletal2020}).

\paragraph{Outline.}  Sec.~\ref{sec:preliminaries} reviews  the notion of \emph{Prob-solvable Loop} and the general theory of \emph{Polynomial Chaos Expansion} (PCE).  Sec.~\ref{sec:PCE_algo} presents our PCE algorithm and the conditions under which it produces accurate approximations to general random functions in probabilistic program loops.  Sec. \ref{sec:genfun} combines general PCE with Prob-solvable loops to automatically compute moments of all orders of state variables. There, we also characterize the structure a probabilistic program ought to have in order to be compatible with the Prob-solvable loops approach for computation of moments.  Sec. \ref{sec:evaluation} demonstrates  the accuracy and  feasibility of our approach on different benchmarks as compared with the state-of-the-art.
 We conclude in Sec. \ref{sec:conclusion}.

\section{Preliminaries}\label{sec:preliminaries}
\subsection{Prob-Solvable Loops}\label{sec:prob-solvable-loops}

\cite{Bartoccietal2019} defined the class of \emph{Prob-solvable loops} for which  moments of all orders of program variables can be computed symbolically: given a Prob-solvable loop and a program variable $x$, their method computes a closed-form solution for $\E(x_n^k)$ for arbitrary $k \in \nat$, where $n$ denotes the $n$th loop iteration.
Prob-solvable loops are restricted to polynomial variable updates.


\begin{definition}[Prob-solvable loops]\label{def:probsolvable}
Let $m \in \nat$ and $x_1,\ldots x_m$ denote real-valued program
variables. 
A Prob-solvable loop with program variables $x_1,\ldots x_m$ is a loop of the form
\begin{equation*}\label{eq:ProbModel}
  I; \texttt{while(true)} \{U\}, 
\end{equation*}
where
\begin{itemize}
    \item $I$ is a sequence of initial assignments over a subset of $\{x_1,\ldots, x_m\}$. The initial values of $x_i$ can be drawn from a known distribution. 
    They can also be real constants.
    \item $U$ is the loop body and is a sequence of $m$ random updates, each of the form:
    \begin{equation*}\label{eq:ProbModel:prob_assignments}
        x_i = \textit{Dist} \quad \text{or} \quad x_i := a x_i + P_{i}(x_1,\dots x_{i-1})
    \end{equation*}
    where $a \in \real$, $P_{i} \in \real[x_1,\ldots,x_{i-1}]$ is a polynomial over program
    variables $x_1,\ldots,x_{i-1}$ and \textit{Dist} is a distribution independent from program variables with computable moments.
\end{itemize}
\end{definition}


Many real-life systems exhibit non-polynomial dynamics and require more general updates, such as, for example, trigonometric or exponential functions.
In this work, we develop a method that allows approximation of non-polynomial assignments in probabilistic loops by polynomial assignments.
In doing so, we can use the methods for Prob-solvable loops to compute the moments for a broader class of stochastic systems.

The programming model we use (Definition~\ref{def:probsolvable}) is a simplified version of the Prob-solvable model as introduced in \cite{Bartoccietal2019}. 
Our approach, described in the following sections, is not limited to this simple fragment of the Prob-solvable and can be used for Prob-solvable loops as originally defined as well as other more general probabilistic loops.
The only requirement is that the loops satisfy the conditions in Section~\ref{sec:conditions}.

\subsection{Polynomial Chaos Expansion}\label{sec:PCE}

Polynomial chaos expansion
recovers a random variable in terms of a linear combination of functionals whose entries are known random variables, sometimes called germs, or, basic variables. 
Let $(\Omega, \Sigma, \pr)$ be a probability space, where $\Omega$ is the set of elementary events, $\Sigma$ is a $\sigma$-algebra of subsets of $\Omega$, and $\pr$  is a probability measure on $\Sigma$. 
Suppose $X$ is a real-valued random variable defined on $(\Omega, \Sigma, \pr)$, such that 
\begin{align}\label{l2fcn}
\E(X^2) &=\int_{\Omega} X^2(\omega) d\pr(\omega) < \infty.
\end{align}
The space of all random variables $X$  satisfying \eqref{l2fcn} is denoted by $L^2(\Omega, \Sigma, \pr)$. That is, the elements of $L^2(\Omega, \Sigma, \pr)$  are real-valued random variables defined on $(\Omega, \Sigma, \pr)$ with finite second moments. If we define  the inner product as
\begin{align}
    \E(XY)&=(X,Y)=\int_{\Omega} X(\omega) Y(\omega) d\pr(\omega) \label{innerprod}
\end{align}
and norm $||X||=\sqrt{\E(X^2)}=\sqrt{\int_{\Omega} X^2(\omega) d\pr(\omega)}$, then $L^2(\Omega, \Sigma, \pr)$ is a Hilbert space; i.e., an infinite dimensional linear space of functions endowed with an inner product and a distance metric.  
 Elements of a Hilbert space can be uniquely specified by their coordinates with respect to an orthonormal basis of functions, in analogy with Cartesian coordinates in the plane. Convergence with respect to $||\cdot||$ is called \textit{mean-square convergence}. A particularly important feature of a Hilbert space is that when the limit of a sequence of functions exists, it belongs to the space.


The elements in $L^2(\Omega, \Sigma, \pr)$ can be classified in two groups:  \textit{basic}  and   \textit{generic} random variables, which we want to decompose using the elements of the first set of basic variables. \cite{Ernstetal2012} showed that the basic random variables that can be used in the decomposition of other functions have finite moments of all orders with continuous probability density functions (pdf). 

The $\sigma$-algebra generated by the basic random variable $Z$ is denoted by $\sigma(Z)$. 
Suppose we restrict our attention to decompositions of a random variable $X=g(Z)$, where $g$ is a function with  $g(Z) \in L^2(\Omega, \sigma(Z),\pr)$ and the basic random variable $Z$ determines the class of orthogonal polynomials $\{\phi_i(Z), i \in \mathbb{N}\}$,
\begin{align}
\left<\phi_i(Z),\phi_j(Z)\right>&= \int_{\Omega} \phi_i(Z(\omega))\phi_j(Z(\omega)) d\pr(\omega) \notag \\
&=\int \phi_i(x)\phi_j(x) f_Z(x) dx=\begin{cases} 1 & i=j \\
0 & i \ne j \end{cases} \label{orth.poly}
\end{align}
which is a polynomial chaos  basis.
If $Z$ is normal with mean zero, the Hilbert space $L^2(\Omega, \sigma(Z),\pr)$ is called \emph{Gaussian} and the related set of polynomials is represented by the family of Hermite polynomials (see, for example, \cite{XiuKarniadakis2002a}) defined on the whole real line. Hermite polynomials form a basis of $L^2(\Omega, \sigma(Z),\pr)$. Therefore, every random variable $X$ with  finite second moment can be approximated  by the truncated PCE 
\begin{align}\label{trunc.PCE}
    X^{(d)}&=\sum_{i=0}^d c_i \phi_i(Z),
\end{align}
for suitable coefficients $c_i$ that depend on the random variable $X$. The truncation parameter $d$ is the highest polynomial degree in the expansion. Since the polynomials are orthogonal,
\begin{align}
    c_i&=\frac{1}{||\phi_i||^2} \left< X,\phi_i \right> =\frac{1}{||\phi_i||^2}\left<g,\phi_i\right> 
    =\frac{1}{||\phi_i||^2} \int_{\real} g(x)\phi_i(x) f_Z(x) dx.
\end{align}
The truncated PCE of $X$ in \eqref{trunc.PCE} converges in mean square to $X$ \cite[Sec. 3.1]{Ernstetal2012}. The first two moments of \eqref{trunc.PCE} are determined by
\begin{align}
    \E(X^{(d)})&= c_0,\\
    \var(X^{(d)})&= \sum_{i=1}^d c_i^2 ||\phi_i||^2.
\end{align}
Representing a random variable by a series of Hermite polynomials in a countable sequence
of independent Gaussian random variables is  known as Wiener–Hermite polynomial chaos expansion.  In applications of Wiener–Hermite PCEs, the underlying Gaussian Hilbert space is
often taken to be the space spanned by a  sequence $\{Z_i, i \in \mathbb{N}\}$ of independent standard Gaussian basic random variables, $Z_i \sim \Ncal(0, 1)$. For computational purposes, the
countable sequence $\{Z_i, i \in \mathbb{N}\}$ is restricted to a finite number $k \in \mathbb{N}$  of random variables. The Wiener–Hermite polynomial chaos expansion converges for 
random variables with finite second moment. Specifically, for any random variable $X \in L^2(\Omega, \sigma(\{Z_i, i \in \mathbb{N}\}), \pr)$, the approximation \eqref{trunc.PCE} satisfies
\begin{align}\label{convergence}
    X_k^{(d)} &\to X \quad \mbox{as } d,k \to \infty
\end{align} 
in mean-square convergence. 
The distribution of $X$ can be quite general; e.g., discrete, singularly continuous, absolutely continuous as well as of mixed type. 

\section{Polynomial Chaos Expansion Algorithm}\label{sec:PCE_algo}
\subsection{Random Function Representation}\label{sec:conditions}


In this section, we state the conditions under which the estimated polynomial is an unbiased and consistent estimator and has exponential convergence rate.\\
Suppose $k$ continuous random variables  $Z_{1},\ldots, Z_{k}$ are used to introduce stochasticity in a PP, with corresponding cumulative distribution functions (cdf) $F_{Z_i}$ for $i=1,\ldots,k$. Also, suppose all $k$ distributions 
have probability density functions, and let  $\Z=(Z_{1},\ldots,Z_{k})$ with cdf $F_{\Z}$. We assume that the elements of $\Z$ satisfy the following conditions:

\begin{itemize}
    \item[(A)] $Z_{i}, i=1,\ldots, k$, are independent.
    \item[(B)] We consider functions $g$ such that $g(\Z) \in L^2(\mathcal{Q}, F_\Z)$, where $\mathcal{Q}$ is the  support of the joint distribution of  $\Z=(Z_{1},\ldots,Z_{k})$ \footnote{$\Omega$ is dropped from the notation as the sample space is not important in our formulation.}.
    \item[(C)] All random variables $Z_{i}$ have distributions that are uniquely defined by their moments. \footnote{Conditions that ascertain this are given in Theorem 3.4 of \cite{Ernstetal2012}.}
\end{itemize}

Under condition (A),  the joint cdf of the components of $\Z$ is $F_{\Z}= \prod_{i=1}^k F_{Z_i}$. 
To ensure the construction of unbiased estimators with optimal exponential convergence rate (see \cite{XiuKarniadakis2002a}, \cite{Ernstetal2012}) in the context of probabilistic loops, we further introduce the following assumptions:

\begin{itemize}
    \item[(D)] $g$ is a function of a fixed number of basic variables (arguments) over all loop iterations.
    \item[(E)] If $\Z(j)=(Z_{1}(j),\ldots,Z_{k}(j))$ is the stochastic argument of $g$ at iteration $j$, then $F_{Z_i(j)}(x) = F_{Z_{i}(l)}(x)$ for all pairs of iterations $(j,l)$ and $x$ in the support of $F_{Z_i}$.
\end{itemize}

If Conditions (D) and (E) are not met, then the polynomial coefficients in the PCE need be computed for each loop iteration individually to ensure optimal convergence rate.
It is straightforward to show the following proposition.

\begin{proposition}\label{propos::func_props}
If functions $f$ and $g$ satisfy conditions (B) and (D), and $\Z=(Z_{1},\ldots, Z_{k_1})$, $\Y=(Y_{1},\ldots, Y_{k_2})$ satisfy conditions (A), (C) and (E) and are mutually independent, then their sum, $f(\Z) + g(\Y)$, and product, $f(\Z) \cdot g(\Y)$, also satisfy conditions (B) and (D).
\end{proposition}

\begin{figure}[t!]
\resizebox{11cm}{11cm}{%
\hspace{-0.5cm}
\begin{tikzpicture}[node distance=2cm]
\node (start) [startstop] {Start};
\node (in1) [io] {\textbf{Input:}
    $\begin{cases}
      Z_{1}, \hspace{0.5cm} \bar{d}_{1},\\
      Z_{2}, \hspace{0.5cm} \bar{d}_{2},\\
      ...\\
      Z_{k}, \hspace{0.5cm} \bar{d}_{k},\\
    \end{cases}$\vspace{0.5cm}
    
Set of random variables and the highest degrees of the corresponding  univariate orthogonal polynomials

};
    
\node (in2) [io, right of=in1, xshift = 4cm] {\textbf{Input}
    $\begin{cases}
      f_{1}(z_{1}), \hspace{0.5cm} \left[a_{1}, b_{1}\right],\\
      f_{2}(z_{2}), \hspace{0.5cm} \left[a_{2}, b_{2}\right],\\
      ...\\
      f_{k}(z_{k}), \hspace{0.5cm} \left[a_{k}, b_{k}\right],\\
    \end{cases}$\vspace{0.5cm}
Set of the density functions and the corresponding supports
};

\node (in3) [io, right of=in2, xshift = 4cm] {\textbf{Input} \\\vspace{0.1cm}
$g(z_{1}, ..., z_{k})$\vspace{0.5cm}

Target function

};

\node (pro1) [process, below of=in2, yshift = -3cm] {\textbf{Orthonormal polynomials:}\\
\vspace{0.2cm}
$\begin{cases}
    \{\bar{p}_{1}^{deg}\}_{deg=0}^{\bar{d}_{1}}\\
    ...\\
    \{\bar{p}_{k}^{deg}\}_{deg=0}^{\bar{d}_{k}}\\
\end{cases}$

};

\node (pro2) [process, below of=in1, yshift = -12.7cm] {
\[
\D^{L \times k} = 
\begin{blockarray}{ccccc}
Z_{1} & Z_{2} & Z_{3} & ... & Z_{k} \\
\begin{block}{(ccccc)}
  0 & 0 & 0 & 0 & 0 \\
  1 & 0 & 0 & 0 & 0 \\
  \vdots & \vdots & \vdots  & \ddots & \vdots\\
  \bar{d}_{1} & \bar{d}_{2} & \bar{d}_{3} & \cdots & \bar{d}_{k} \\
\end{block}
\end{blockarray}
 \]
};

\draw [arrow, line width=1mm, color = green!50, text = black] (in2) -- node[anchor=east] {\begin{center}
    \textbf{1: Gram-Schmidt Process}
\end{center}} (pro1);
\draw [arrow, line width=1mm, color = green!90] (in1) |-  (pro1);

\begin{scope}[transform canvas={xshift=-1cm}]
  \draw [arrow, line width=1mm, color = brown!90, text = black] (in1) -- node[anchor=east] {\vspace{5cm}
\begin{center}
    \textbf{2:  Matrix of\\  polynomials' combinations}
\end{center}: } (pro2);
\end{scope}

\node (pro3) [process, below of=in3, yshift = -6.5cm] {\textbf{Fourier coefficients}:\\
$c_{j} = \int\limits_{a_{1}}^{b_{1}}...\int\limits_{a_{k}}^{b_{k}}g(z_{1}, ...z_{k})\prod\limits_{i=1}^{k}\left[f_{Z_{i}}(z_{i})\bar{p}^{d_{ji}}_{i}(z_{i})\right]dz_{k}...d_{z_{1}}$

};

\draw [arrow, line width=1mm, color = blue!90, text = black] (in3) --  node[anchor=west] {\vspace{3cm}
\begin{center}
    \textbf{3: \textbf{Calculation of\\ coefficients}}
\end{center} } (pro3);

\begin{scope}[transform canvas={xshift=1cm}]
  \draw [arrow, line width=1mm, color = blue!90, text = black, shorten <= -1cm] (pro2) -| (pro3);
\end{scope}

\draw [arrow, line width=1mm, color = blue!90, text = black] (pro1) -- (pro3);

\node (out1) [io, below of=pro1, yshift=-5cm] {\textbf{Output:} \\
\vspace{0.2cm}
$\sum\limits_{j} c_{j}\prod\limits_{i} \bar{p}_{i}^{d_{ji}}$
};

\draw [arrow, line width=1mm] (pro3) |- (out1);

\draw [arrow, line width=1mm] (pro1) -- node[anchor=west] {\vspace{3cm}
\begin{center}
    \textbf{4: \textbf{Summation of weighted polynomials}}
\end{center} } (out1);
\end{tikzpicture}
}
\caption{Illustration of PCE algorithm}
\label{fig::PCE_Algo}
\end{figure}
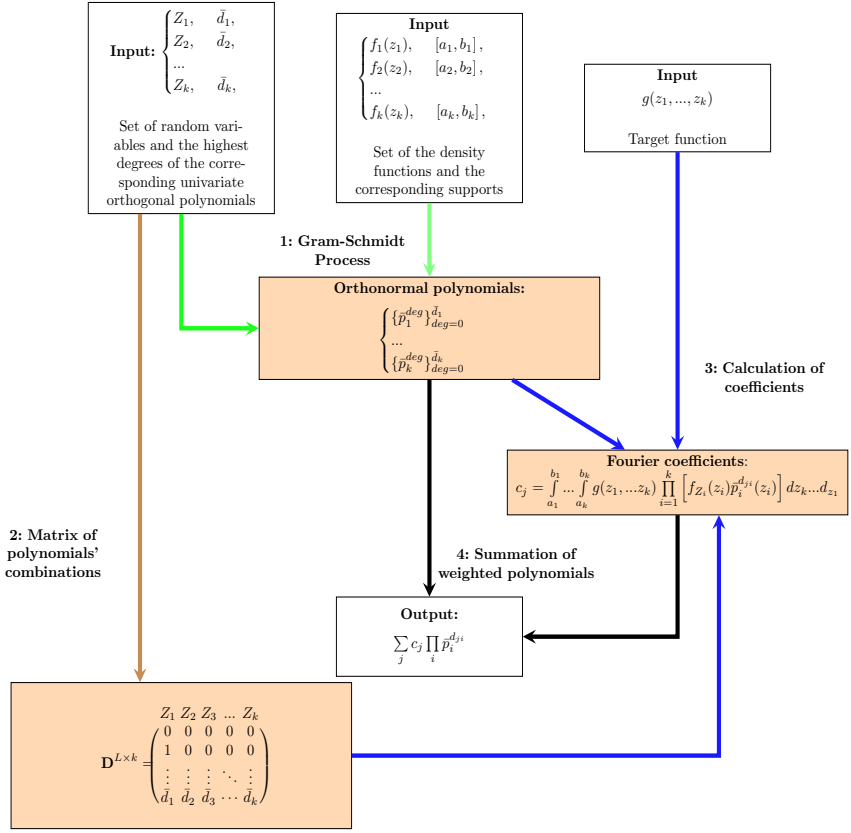

\subsection{PCE Algorithm}\label{sec:PCEalgo}

Let $Z_1,\ldots,Z_k$ be independent continuous random variables, with respective cdfs $F_i$, satisfying conditions (A), (B) and (C), and $\Z=(Z_1,\ldots,Z_k)^T$ with cdf $\F = \prod_{i=1}^k F_{i}$ and support $\mathcal{Q}$.
The function $g: \real^{k} \to \real$, with $g \in L^2(\mathcal{Q}, \F)$ can be approximated  with the truncated orthogonal polynomial expansion, as described in  Fig. \ref{fig::PCE_Algo}, 
\begin{align}\label{polyexpanse}
g(\Z) &\approx \hat{g}(\Z) = \sum_{\substack{d_i \in \{0,\ldots,\bar{d}_i\},\\ i=1,\ldots, k}} c(d_1,\ldots,d_k) z_1^{d_1}  \cdots z_k^{d_k} 
= \sum_{j=1}^L c_{j}\prod\limits_{i=1}^{k}\bar{p}_{i}^{d_{ji}}(z_{i}), 
\end{align}
where 
\begin{itemize}
 \item $\bar{p}_{i}^{d_{ji}}(z_{i})$ is a polynomial of degree $d_{ji}$, and belongs to the set of orthogonal polynomials with respect to $F_{Z_i}$ that are calculated with the Gram-Schmidt orthogonalization procedure\footnote{Generalized PCE typically entails using orthogonal basis polynomials specific to the distribution of the basic variables, according to the Askey scheme of \cite{XiuKarniadakis2002a,Xiu2010}. We opted for the most general procedure that can be used for any basic variable distribution.};
    \item $\bar{d}_{i} = \max\limits_{j}(d_{ji})$ is the highest degree of the univariate orthogonal polynomial, for $i=1,\ldots,k$; 
    \item $L = \prod\limits_{i=1}^{k}(1 + \bar{d}_{i})$ is the total number of multivariate orthogonal polynomials and equals the truncation constant;
    \item $c_{j}$ are the Fourier coefficients.
\end{itemize}

The Fourier coefficients are calculated using
\begin{align}\label{coef}
c_{j} &= \int\limits_{\mathcal{Q}}g(z_{1},...,z_{k})p_{i}^{d_{ji}}(z_{i})d\F = 
\idotsint\limits_{\mathcal{Q}} g(z_{1},...,z_{k})\prod\limits_{i=1}^{k}\bar{p}_{i}^{d_{ji}}(z_{i})dF_{Z_k}...dF_{Z_1},
\end{align}
by Fubini's theorem. An example with calculation details of the coefficients in~\eqref{coef} is given in Appendix \ref{app:B}.

\begin{example}
Returning to the Turning Vehicle model in Fig.~\ref{fig:turningexample}, 
the non-polynomial functions to approximate are $g_1 = cos$ and $g_2 = sin$ from the updates of program variables $y, x$, respectively. In both cases, we only need to consider a single basic random variable, $Z \sim \Ncal(0,0.01)$ ($\psi$ in Fig.~\ref{fig:turningexample}). 

For the approximation, we use polynomials of degree up to $5$. 
Eqn. \eqref{polyexpanse} has the following form for the two functions,
\begin{equation}\label{eqn:g_1}
    \hat{g}_1(z) = cos(\psi) = a_0 + a_1\psi + ... + a_5\psi^5
\end{equation}
and
\begin{equation}\label{eqn:g_2}
    \hat{g}_2(z) = sin(\psi) = b_0 + b_1\psi + ... + b_5\psi^5.
\end{equation}
We compute the coefficients $a_i, b_i$ in equations~\eqref{eqn:g_1}-\eqref{eqn:g_2} using~\eqref{coef} to obtain the values shown in Fig.~\ref{fig:turningexample}.
\end{example}

\paragraph{Complexity.} 

Assuming the expansion is carried out up to the same polynomial degree $d$ for each basic variable, $\bar{d}_{i} = d$, $\forall i=1,...,k$. This implies $d=\sqrt[k]{L}-1$. The complexity of the scheme is $\mathcal{O}(sd^{2}k + s^{k}d^{k})$, where $\mathcal{O}(s)$ is the complexity of computing univariate integrals.

The complexity of our approximation scheme is comprised of two parts:
(1) the orthogonalization process and (2) the calculation of coefficients. Regarding (1), we orthogonalize and normalize $k$ sets of $d$  basic linearly independent polynomials during the Gram-Schmidt process.
For degree $d{=}1$, we need to calculate one integral, the inner product with the previous polynomial.
Additionally, we need to compute one more integral, the norm of itself (for normalization).
For each subsequent degree $d'$, we must calculate $d'$ additional new integrals. The computation of each integral has complexity $\mathcal{O}(s)$.
Regarding (2), the computation of the coefficients requires  calculating $L{=}(d{+} 1)^{k}$ integrals with $k$-variate functions as integrands.

\medskip

We define the approximation error  to be
\begin{equation}\label{error_std}
    se(\hat{g}) = \sqrt{\int\limits_{\mathcal{Q}}\left(g(z_{1},...,z_{k}) - \hat{g}(z_{1},...,z_{k})\right)^{2}dF_{Z_1}\ldots dF_{Z_k}}
\end{equation}
since $\E(\hat{g}(Z_{1},...,Z_{k}))=g(Z_1,\cdots,Z_k)$ by construction.

Appendix \ref{app:C} contains  examples of polynomial chaos expansion for functions of up to three variables.
The implementation of this algorithm may become challenging when the random functions  have complicated forms and  the number of parametric uncertainties is large. In this case, the calculation of the PCE coefficients involves high dimensional integration, which may prove difficult and time prohibitive for real-time applications \cite{SonDu2020}.

\section{Prob-Solvable Loops for General Non-Polynomial Functions}\label{sec:genfun}

PCE allows incorporating non-polynomial updates into Prob-solvable loop programs and use the algorithm in \cite{Bartoccietal2019} and exact tools, such as \emph{Polar}~\cite{Moosbruggeretal2022}, for moment (invariant) computation.
We identify two classes of programs based on how the distributions of the random variables generated by the programs vary.

\subsection{	Iteration-Stable Distributions of Random Arguments}\label{sec:stable} 

Let $\Pcal$ be an arbitrary Prob-solvable loop and suppose that a (non-basic) state variable $x \in \Pcal$ has a non-polynomial $L^2$-type update $g(\Z)$, where $\Z = (Z_{1}, ..., Z_{k})^{T}$ is a vector of (basic) continuous, independent, and identically distributed random variables \textit{across iterations}. That is, if $f_{Z_j(n)}$ is the pdf  of the random variable $Z_j$ in iteration $n$, then  $f_{Z_j( n)} \equiv f_{Z_j(n')}$, for all iterations $n, n'$ and $j=1,\ldots,k$. The basic random variables $Z_1,\ldots, Z_k$ and the update function $g$ satisfy conditions (A)--(E) in Section~\ref{sec:conditions}.
For the class of Prob-solvable loops where all variables with non-polynomial updates satisfy these conditions,  the computation of the Fourier coefficients in the PCE approximation \eqref{polyexpanse} can be carried out as explained in Section~\ref{sec:PCEalgo}.
In this case, the convergence rate is optimal.

\subsection{Iteration Non-Stable Distribution of Random Arguments}\label{sec:non-stable} 

Let $\Pcal$ be an arbitrary Prob-solvable loop and suppose that a state variable $x \in \Pcal$ has a non-polynomial $L^2$-type update $g(\Z)$, where $\Z = (Z_{1}, ..., Z_{k})^{T}$ is a vector of continuous independent  but \textit{not necessarily identically} distributed random variables across iterations. 
For this class of Prob-solvable loops, conditions (A)--(C) in Section~\ref{sec:conditions} hold, but (D) and/or (E) may not be fulfilled.
In this case, we can ensure optimal exponential convergence by fixing the number of loop iterations.
For unbounded loops, we describe an approach converging in mean-square and establish its convergence rate next.

\paragraph{Conditional estimator given number of iterations.}

Let $N$ be an a priori fixed finite integer, representing the maximum iteration number. The set $\{1, ..., N\}$ is a finite sequence of iterations for the Prob-solvable loop $\Pcal$. 

Iterations are executed sequentially for $n =1,\ldots, N$, which allows  the estimation of  the final functional that determines the target state variable at each iteration $n \in \{1, ..., N\}$ and its set of supports. Knowing these features, we can carry out $N$ successive expansions. Let $P(n)$ be a PCE of $g(\Z)$ for iteration $n$. We introduce an additional program variable $c$ that counts the loop iterations.
The variable $c$ is initialized to $0$ and incremented by $1$ at the beginning of every loop iteration.
The final estimator of  $g(\Z)$ can be represented as 
\begin{equation}\label{eq:cond-estimator}
    \hat{g}(\Z) = \sum\limits_{n = 1}^{N}P(n)\left[\prod\limits_{j = 1, j \neq n}^{N}\frac{(c - j)}{n - j}\right].
\end{equation}
Replacing non-polynomial functions with \eqref{eq:cond-estimator} results in a program with only polynomial-type updates and \textit{constant} polynomial structure; that is, polynomials with  coefficients that remain constant across iterations. Moreover, the estimator is unbiased with optimal exponential convergence on the set of iterations $\{1, ..., N\}$ \cite{XiuKarniadakis2002a}.

\paragraph{Unconditional estimator.}\label{sec:uncond} 
Here the iteration number is unbounded. Without loss of generality, we consider a single basic random variable $Z$; that is, $k{=}1$.  The function  $g(Z)$ is scalar valued and can be represented as a polynomial of \textit{nested} $L^2$ functions, which depend on polynomials of the argument variable. Each nested functional argument is expressed as a sum of orthogonal polynomials yielding  the final estimator, which is itself a polynomial.

Since PCE converges to the function it approximates in mean-square (see \cite{Ernstetal2012}) 
on the whole interval (argument's support),  PCE converges on any sub-interval of the support of the argument in the same sense.


Let us consider a function $g$ with sufficiently large domain, and a random variable $Z$ with known distribution and support. For example, $g(Z) = e^{Z}$, with $Z \sim N(\mu, \sigma^{2})$. The domain of $g$ and the support of $Z$ are the real line.
We can expand $g$ into a PCE with respect to the distribution of  $Z$ as 
\begin{align}\label{truef}
 g(Z) &= \sum\limits_{i=0}^{\infty}c_{i}p_{i}(Z).
\end{align}
The distribution of $Z$ is reflected in the polynomials in \eqref{truef}. Specifically,  $p_{i}$, for $i=0,1,\ldots$, are  Hermite polynomials of special type in that they are orthogonal (orthonormal) with respect to $N(\mu, \sigma^{2})$. They also form an  orthogonal basis of the space of $L^2$ functions. Consequently, any function in $L^2$ can be estimated arbitrarily closely by these polynomials.
In general, any continuous distribution with finite moments of all orders and sufficiently large support can also be used as a model for basic variables in order to construct a basis for $L^2$ (see \cite{Ernstetal2012}).

Now suppose that the distribution of the  underlying variable $Z$ is unknown with pdf $f(Z)$ that is continuous on its support $\left[a, b \right]$. Then, there exists another basis of polynomials, $\{q_{i}\}_{i=0}^{\infty}$, which are orthogonal on the support $\left[a, b \right]$ with respect to the pdf $f(Z)$. Then, on the interval $\left[a, b \right]$,  $g(Z) = \sum_{i=0}^{\infty}k_{i}q_{i}(Z)$, and $\mathbb{E}_{f}\left[g(Z)\right] = \mathbb{E}_{f}\left[\sum_{i=0}^{M}k_{i}q_{i}(Z) \right]$, $\forall M \geq 0$.

Since  $\left[a, b \right] \subset \real$,  the expansion $\sum_{i=0}^{\infty}c_{i}p_{i}(Z)$ converges  in mean-square to $g(Z)$ on $\left[a, b \right]$. In the limit, we have  $g(Z) = \sum_{i=0}^{\infty}c_{i}p_{i}(Z)$ on the interval $\left[a, b \right]$. Also, $\mathbb{E}_{f}\left(g(Z)\right) = \mathbb{E}_{f}\left(\sum_{i=0}^{\infty}c_{i}p_{i}(Z)\right)$ for the true pdf $f$ on $\left[a, b \right]$.
In general, though, it is not true that $\mathbb{E}_{f}\left(g(Z)\right) = \mathbb{E}_{f}\left(\sum_{i=0}^{M}c_{i}p_{i}(Z)\right)$ for any arbitrary $M \geq 0$ and any pdf $f(Z)$ on $\left[a, b \right]$, as the estimator is biased.

To capture this discrepancy, we define the approximation error as
\begin{align}\label{approxer}
  e(M) &= \mathbb{E}_{f}\left[ g(Z) -  \sum\limits_{i=0}^{M}c_{i}p_{i}(Z) \right]^{2} = \mathbb{E}_{f}\left[ \sum\limits_{i=M+1}^{\infty}c_{i}p_{i}(Z) \right]^{2}.
\end{align}
\paragraph{Computation of error bound.}

Assume the true pdf $f_Z$ of $Z$ is supported on $\left[ a, b\right]$. Also, assume the domain of $g$ is $\mathbb{R}$. The random function $g(Z)$ has PCE  on the whole real line based on Hermite polynomials $\{p_{i}(Z)\}_{i=0}^{\infty}$ that are orthogonal with respect to the standard normal pdf $\phi$. The truncated  expansion estimate of \eqref{truef} with respect to a normal basic random variable is
\begin{align}
    \hat{g}(Z)&=  \sum\limits_{i=0}^{M}c_{i}p_{i}(Z). \label{fapprox}
\end{align} 
We compute an upper-bound for the approximation error for our scheme in Theorem~\eqref{err-bound}.

\begin{theorem}\label{err-bound}
Suppose $Z$ has density $f$ supported on $[a,b]$, $g: \real \to \real$ is in $L^2$, and $\phi$ denotes the standard normal pdf. Under \eqref{truef} and \eqref{fapprox},
\begin{align}
   \left \| g(Z) - \hat{g}(Z) \right \|_{f}^{2}&= \int_a^b \left(g(z)-\hat{g}(z)\right)^2 f_Z(z) dz \notag \\
   & \qquad \le   \left( \frac{2}{\min{(\phi(a), \phi(b))}} + 1 \right) \var_{\phi}\left(g(Z)\right). \label{bound}
\end{align}
\end{theorem}
The proof is given in the Appendix.

\medskip
The upper bound in \eqref{bound} depends only on the support of $f$, the pdf of~$Z$, and the function $g$. 
If $Z$ is standard normal ($f=\phi$), then the upper bound in \eqref{bound} equals $\var_\phi(g(Z))$.

\begin{remark}
The approximation error inequality in \cite[Lemma 1]{Muehlpfordtetal2018}, 
\begin{equation}\label{muelpford}
\left \|{ g(Z) - \sum _{i=0}^{T} c_{i}  p_{i}(Z) }\right \| \leq \frac {\| g(Z)^{(k)} \|}{ \prod _{i=0}^{k-1} \sqrt {T-i+1}}, 
\end{equation}
is a special case of  Theorem \ref{err-bound} when $Z \sim \mathcal{N}(0,1)$ and $f=\phi$, and the  polynomials $p_i$ are Hermite. In this case, the left hand side of \eqref{muelpford} equals  
$\sqrt{\sum_{i=n+1}^{\infty}c_{i}^{2}}$.
\end{remark}

Although Theorem~\ref{err-bound} is restricted to distributions with bounded support, the approximation in  \eqref{fapprox}  also converges for distributions with unbounded support.

\section{Evaluation}\label{sec:evaluation}

In this section, we evaluate our approach on four benchmarks from the literature.
We use our method based on PCE to approximate non-polynomial functions.
After PCE, all benchmark programs fall into the class of Prob-solvable loops.
We use the static analysis tool \textit{Polar}~\cite{Moosbruggeretal2022} on the resulting Prob-solvable loops to compute the moments of the program variables parameterized by the loop iteration $n$.
All experiments were run on a machine with 32 GB of RAM and a 2.6 GHz Intel i7 (Gen 10) processor.

\begin{figure}[!t]
\centering
  \includegraphics[scale=.55]{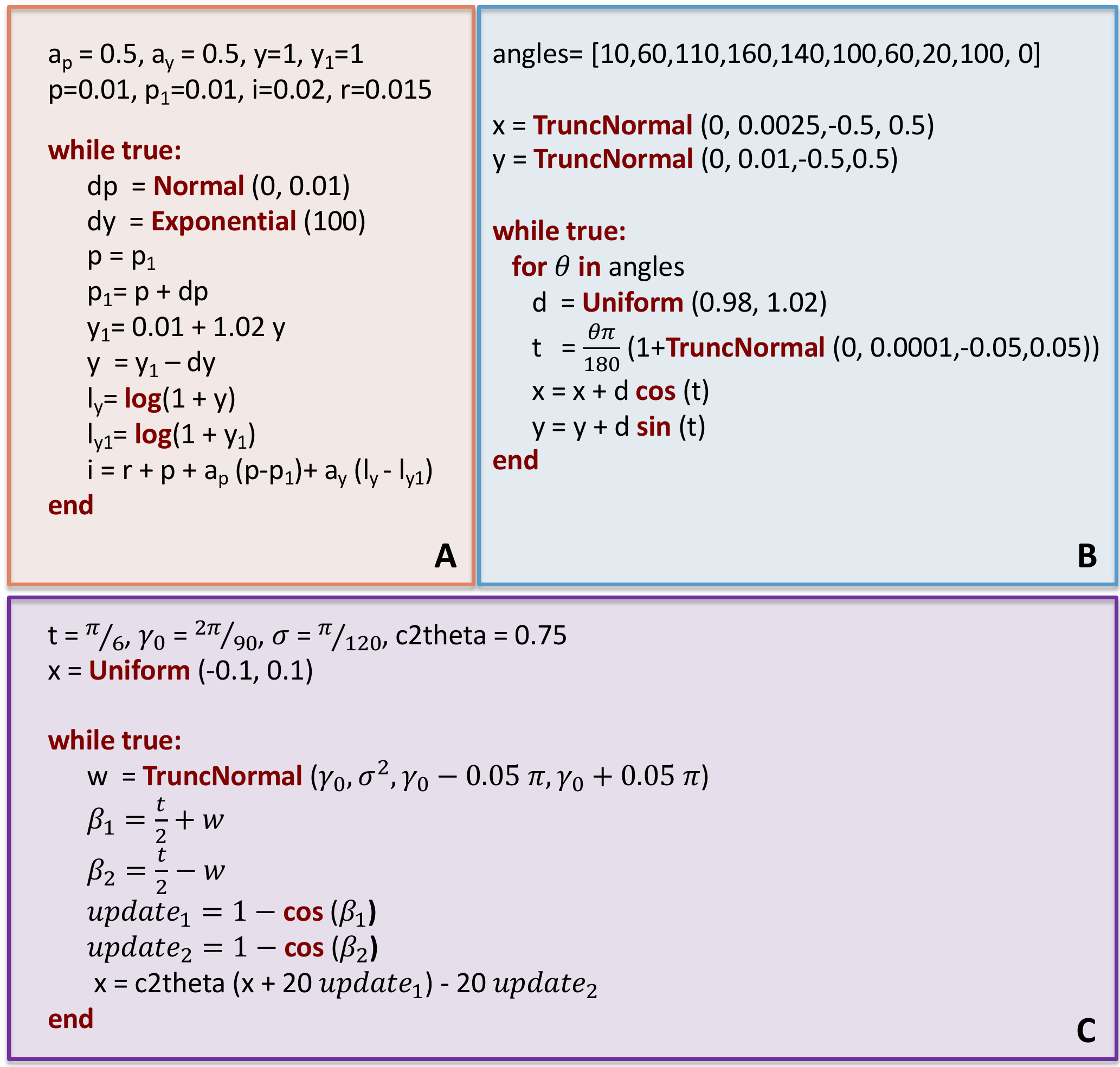}
  \caption{Probabilistic loops: (A) Taylor rule~\cite{Taylor1993}, (B) 2D Robotic Arm~\cite{Bouissouetal2016} (in the figure we use the inner loop as syntax sugar to keep the program compact), (C) Rimless Wheel Walker~\cite{SteinhardtT12}. } 
  \label{fig:otherexamples}
\end{figure}

\bigskip

{\bf Taylor Rule Model.}
Central banks set monetary policy by raising or lowering their target for the federal funds rate.
The Taylor rule\footnote{It was proposed by the American economist John B. Taylor as a technique to stabilize economic activity by setting an interest rate \cite{Taylor1993}.} is an equation intended to describe the interest rate decisions of central banks.
The  rule relates the target of the federal funds rate to the current state of the economy through the formula
\begin{equation*}
    i_{t} = r^{*}_{t} + \pi_{t} + a_{\pi}(\pi_{t} - \pi^{*}_{t}) + a_{y}(y_{t} - \bar{y}_{t}),
\end{equation*}
where $i_{t}$ is the nominal interest rate, $r^{*}_{t}$ is the equilibrium real interest rate, $r^{*}_{t} = r$, $\pi_{t}$ is inflation rate at  $t$, $\pi^{*}_{t}$ is the short-term target inflation rate at $t$, $y_{t} = \log(1 + Y_{t})$, with $Y_{t}$  the real GDP, and $\bar{y}_{t} = \log(1 + \bar{Y}_{t})$, with $\bar{Y}_{t}$ denoting the potential real output. 

Highly-developed economies grow exponentially with a sufficiently small rate (e.g., according to the World Bank,\footnote{\url{https://data.worldbank.org/indicator/NY.GDP.MKTP.KD.ZG?locations=US}} the average growth rate of the GDP in the USA in 2001-2020 equals to 1,73\%).
Therefore, we set the growth rate of the potential output to 2\%.
Moreover, we follow \cite{Atkeson_Ohanian_2001} and model inflation as a martingale process; that is, $\E_{t}\left[\pi_{t+1}\right] = \pi_{t}$.
The Taylor rule model is described by the program in Fig.~\ref{fig:otherexamples}, A.

Fig. \ref{fig:aside} illustrates the performance of our approach as a function of the polynomial degree of our approximation. The approximations to the true first moment (in red) are plotted in the left panel and the relative errors for the first and second moments are in the middle and the right panels, respectively, over iteration number. The $y$-axis in both middle  and right panels shows relative errors calculated as $rel.err = |est - true|/true.$
All plots show that the approximation error is low and that it deteriorates as the polynomial degree increases from 3 to 9, across iterations.  The drop is sharper for the second moment.

\begin{table}[!t]
\resizebox{5.0in}{0.8in}{
        \centering
        \begin{tabular}{ccc}
            \includegraphics[scale=.8]{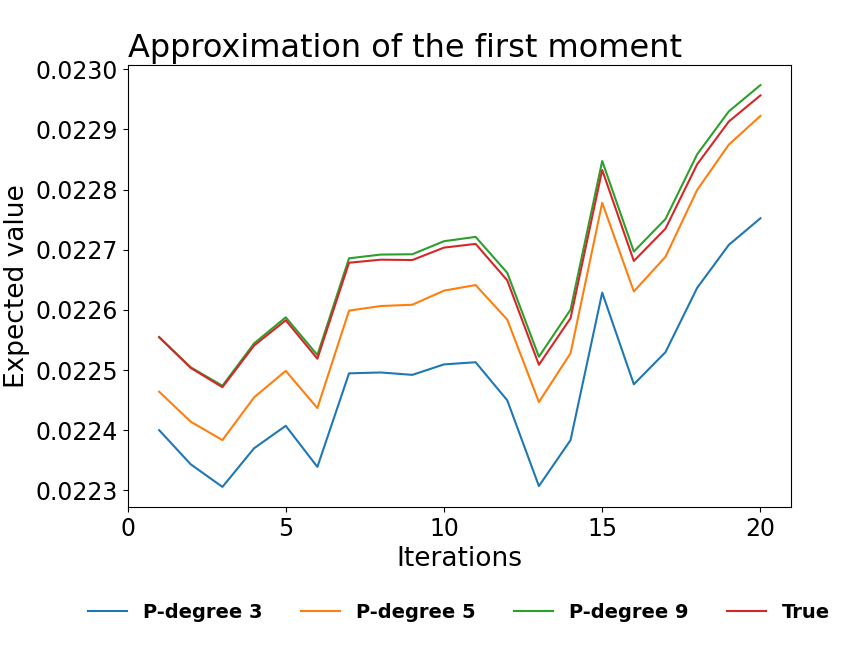} & \includegraphics[scale=.8]{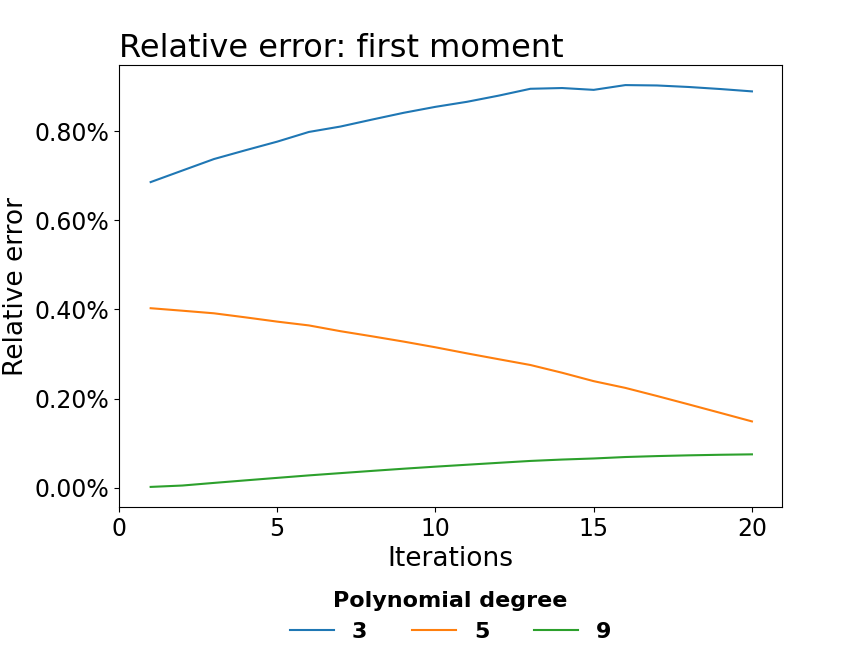} & 
            \includegraphics[scale=.8]{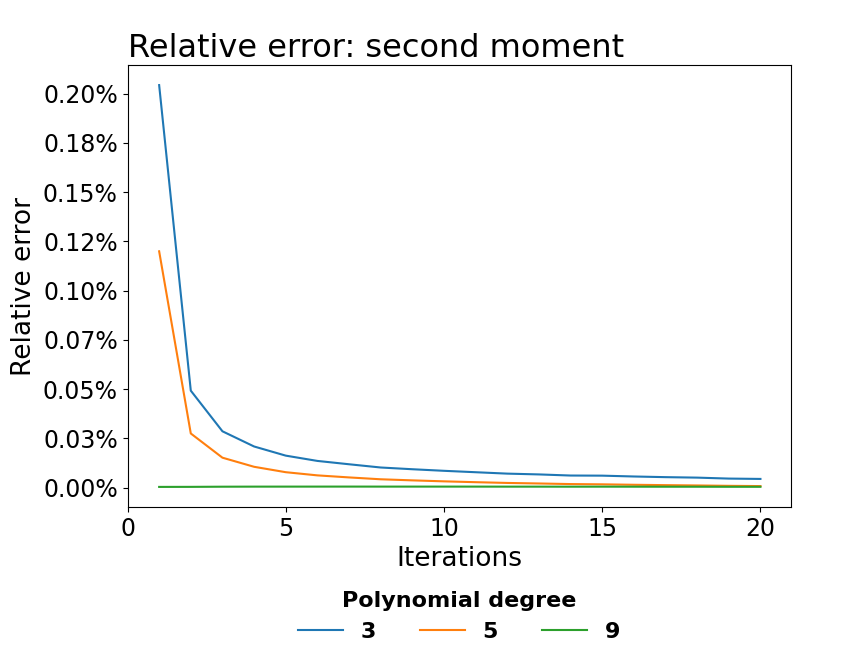} 
        \end{tabular}

}
        
        \captionof{figure}{The approximations and their relative errors for the Taylor rule model.} 
        \label{fig:aside}

    \end{table}

\bigskip

{\bf Turning Vehicle Model.}
The Turning vehicle model is described by the program in Fig.~\ref{fig:turningexample}.
The model was introduced in \cite{Srirametal2020} and depicts the position of a vehicle, as follows. 
The state variables are $(x, y, v, \psi)$, where $(x, y)$ is the vehicle's position with velocity $v$ and yaw angle $\psi$.
The vehicle’s velocity is stabilized around $v_{0} = 10$ m/s.
The dynamics are modelled by the equations  $x(t{+}1) = x(t) + \tau  v  \cos(\psi(t))$, $y(t{+}1) = y(t) + \tau v \sin(\psi(t))$, $v(t{+}1) = v(t) + \tau (K(v(t) - v_{0}) + w_{1}(t{+}1))$, and $\psi(t{+}1) = \psi(t) + w_{2}(t{+}1)$.
The disturbances $w_{1}$ and $w_{2}$ have distributions $w_{1} \sim U[-0.1, 0.1]$, $w_{2} \sim N(0, 0.1)$.
Moreover, as in \cite{Srirametal2020}, we set $K = -0.5$.
Initially, the state variables are distributed as:
$x(0) \sim U[-0.1, 0.1]$, $y(0) \sim U[-0.5, -0.3]$, $v(0) \sim U[6.5, 8.0]$, $\psi(0) \sim N(0, 0.01)$.
We  allow all normally distributed parameters take values over the entire real line, in contrast to \cite{Srirametal2020} who could not accommodate distributions with infinite support and required the normal variables to take values over finite intervals. 




This program requires the approximation of trigonometric functions for the computation of the location of the vehicle at time $t$. We used PCEs of degree 3, 5 and 9, built on  a basic standard normal random variable, to estimate the dynamics of $x$ and $y$, the first and second coordinates of the vehicle's location. For all three PCEs for $\sin$, the Prob-solvable loops tool estimates the first moment of $y$ to be the same, namely $-2/5$. We report the value of the first moment of $x$ in Table~\ref{tab:LiOModl}. The \textit{polynomial form} of \cite{Srirametal2020} can not be applied to approximate any moments of $x$ and $y$. We see that our PCE based estimate is very close to the ``true'' first moment of $x$, with the 9th degree PCE being the closest, as expected.

\bigskip

{\bf Rimless Wheel Walker.}
The \emph{Rimless wheel walker}~\cite{Srirametal2020,SteinhardtRuss2012} is a system that describes a human walking.
The system models a rotating wheel consisting of $n_s$ spokes, each of length $L$, connected at a single point.
The angle between consecutive spokes is $\theta=2\pi / n_s$.
We set $L = 1$ and $\theta = \pi/6.$
The Rimless wheel walker is modeled by the program in Fig.~\ref{fig:otherexamples}, C.
For more details we refer to~\cite{Srirametal2020}.


\bigskip

{\bf Robotic Arm Model.}
Proposed and studied in \cite{Bouissouetal2016,Sankaranarayanan2020,Srirametal2020}, this system models the position of a 2D robotic arm.
The arm  moves through translations and rotations.
Moreover, at every step, errors in movement are modeled with probabilistic noise.
The Robotic arm model is described by the program in Fig.~\ref{fig:otherexamples},~B.


\begin{table}[t!]
\centering
\footnotesize
\renewcommand{\arraystretch}{1.1}
\setlength\abovecaptionskip{8pt}
\setlength{\tabcolsep}{3pt}
\begin{tabular}{@{}lcccccc@{}}

\rotatebox{0}{\textbf{Benchmark}} & \rotatebox{0}{\textbf{Target}} & \rotatebox{0}{\textbf{Poly form}} & \rotatebox{0}{\textbf{Sim.}} & \rotatebox{0}{\textbf{Deg.}} & \rotatebox{0}{\textbf{Result}} & \rotatebox{0}{\showclock{0}{0} \textbf{Runtime}} \\ 
\toprule

\begin{tabular}{@{}l@{}} Taylor rule \\ model \end{tabular} &
\begin{tabular}{@{}c@{}} $\E \left(i_n\right)$ \\ $n{=}20$ \end{tabular} & 
$\times$ &
0.02298 &
\begin{tabular}{@{}c@{}} $3$ \\ $5$ \\ $9$ \end{tabular} & 
\begin{tabular}{@{}c@{}} $0.02278$ \\ $0.02295$ \\ $0.02300$ \end{tabular} &
\begin{tabular}{@{}c@{}}  $0.4s{+}0.5s$ \\  $0.5s{+}5.0s$ \\  $5.9s{+}34.6s$ \end{tabular} \\ 
\midrule

\begin{tabular}{@{}l@{}} Turning vehicle \\ model \end{tabular} &
\begin{tabular}{@{}c@{}} $\E \left(x_n\right)$ \\ $n{=}20$ \end{tabular} &
$\times$ &
15.69792 &
\begin{tabular}{@{}c@{}} $3$ \\ $5$ \\ $9$ \end{tabular} &
\begin{tabular}{@{}c@{}} $14.44342$ \\ $15.43985$ \\ $15.60595$ \end{tabular} &
\begin{tabular}{@{}c@{}} $0.6s{+}3.6s$ \\ $1.4s{+}9.2s$ \\  $15.6s{+}16.1s$ \end{tabular} \\ 
\midrule

\begin{tabular}{@{}l@{}} Turning vehicle \\ model (trunc.) \end{tabular} &
\begin{tabular}{@{}c@{}} $\E \left(x_n\right)$ \\ $n{=}20$ \end{tabular} &
\begin{tabular}{@{}c@{}} for deg. 2 \\ $\left[-3 \cdot 10^5, 3 \cdot 10^5\right]$ \\ \showclock{0}{37} 1117s \end{tabular} &
15.69882 &
\begin{tabular}{@{}c@{}} $3$ \\ $5$ \\ $9$ \end{tabular} &
\begin{tabular}{@{}c@{}} $14.44342$ \\ $15.43985$ \\ $15.60595$ \end{tabular} &
\begin{tabular}{@{}c@{}} $0.6s{+}3.6s$ \\ $1.4s{+}9.1s$ \\  $15.6s{+}19.0s$ \end{tabular} \\
\midrule

\begin{tabular}{@{}l@{}} Rimless wheel \\ walker \end{tabular} &
\begin{tabular}{@{}c@{}} $\E \left(x_n\right)$ \\ $n{=}2000$ \end{tabular} &
\begin{tabular}{@{}c@{}} for deg. 2 \\ $\left[1.791,1.792\right]$ \\ \showclock{0}{46} 46.1s \end{tabular} &
1.79155 &
\begin{tabular}{@{}c@{}} $1$ \\ $2$ \\ $3$ \end{tabular} &
\begin{tabular}{@{}c@{}} $1.79159$ \\ $1.79159$ \\ $1.79159$ \end{tabular} &
\begin{tabular}{@{}c@{}}  $0.2s{+}0.5s$ \\ $0.3s{+}0.5s$ \\  $0.6s{+}0.5s$ \end{tabular} \\ 
\midrule

\begin{tabular}{@{}l@{}} Robotic arm \\ model \end{tabular} &
\begin{tabular}{@{}c@{}} $\E \left(x_n\right)$ \\ $n{=}100$ \end{tabular} &
\begin{tabular}{@{}c@{}} for deg. 2 \\ $\left[268.87,268.88\right]$ \\ \showclock{0}{36} 35.8s \end{tabular} &
268.853 &
\begin{tabular}{@{}c@{}} $1$ \\ $2$ \\ $3$ \end{tabular} &
\begin{tabular}{@{}c@{}} $268.85236$ \\ $268.85227$ \\ $268.85227$ \end{tabular} &
\begin{tabular}{@{}c@{}}  $1.3s{+}0.3s$ \\  $2.5s{+}0.6s$ \\  $4.8s{+}0.7s$ \end{tabular} \\ 
\bottomrule
\end{tabular}
\caption{
Evaluation of our approach on $4$ benchmarks.
Poly form = the interval for the target as reported in \cite{Srirametal2020};
Sim = target approximated through $10^6$ samples;
Deg. = maximum degrees used for the approximation of the non-linear functions;
Result = result of our method per degree;
Runtime = execution time of our method in seconds (time of PCE + time of \textit{Polar}).
}
\label{tab:LiOModl}
\end{table}

\medskip

The \emph{Rimless wheel walker} and the \emph{Robotic arm model} are the only two benchmarks from \cite{Srirametal2020} containing non-polynomial updates.
In \cite{Srirametal2020}, polynomial forms of degree 2 were used to compute bounding intervals for $\E(x_n)$ (for fixed $n$) for the \emph{Rimless wheel walker} and the \emph{Robotic arm model}.
Their tool does not support the approximation of logarithms (required for the \emph{Taylor rule model}) and distributions with unbounded support (required for the \emph{Turning vehicle model}).
To facilitate comparison to polynomial forms, our set of benchmarks is augmented with a version of the \emph{Turning vehicle model} using truncated normal distributions instead of normal distributions with unbounded support (\emph{Turning vehicle model (trunc.)} in Table~\ref{tab:LiOModl}).
We note that the technique in \cite{Srirametal2020} supports more general probabilistic loops than  Prob-solvable loops.
However, as already mentioned in Section~\ref{sec:prob-solvable-loops}, we emphasize that our results in Sections~\ref{sec:PCE}-\ref{sec:genfun} are not limited to Prob-solvable loops and can be applied to approximate non-linear dynamics for more general probabilistic loops.

Table~\ref{tab:LiOModl} summarizes the evaluation of our approach on these four benchmarks and of the technique based on \textit{polynomial forms} of \cite{Srirametal2020} on the directly comparable \emph{Turning vehicle model (trunc.)}, \emph{Rimless wheel walker} and the \emph{Robotic arm} models. 
Our results illustrate that our method is able to accurately approximate general non-linear dynamics for challenging programs. Specifically, for the \emph{Rimless wheel walker} model, our first moment estimate is reached upon with a first degree approximation, is close to the truth up to the fourth decimal and falls in the interval estimate of the polynomial forms technique. For the \emph{Robotic arm model}, our results lie outside the interval predicted by the polynomial forms technique, yet are closer to the simulation (``truth'') calculated with $10^6$ samples.
Moreover, our simulation agrees with the estimation provided in \cite{Srirametal2020}.

Our experiments also demonstrate that our method provides suitable approximations in a fraction of the time required by the technique based on polynomial forms.
While polynomial forms additionally provide an error interval, they need to be computed on an iteration-by-iteration basis.
In contrast, our method based on PCE and Prob-solvable loops computes an expression for the target parameterized by the loop iteration $n \in \mathbb{N}$ (cf. Fig.~\ref{fig:turningexample}).
As a result, increasing the target iteration does \emph{not} increase the runtime of our approach.

Both \emph{Robotic arm} and  \emph{Rimless wheel walker} models contain no stochastic accumulation: each basic random variable is iteration-stable and  can be estimated using the scheme in Section~\ref{sec:stable}.
Therefore, for these two benchmarks, our estimation converges exponentially to the true values. 
On the other hand, the \emph{Taylor rule model} and the \emph{Turning vehicle model} contain stochasticity accumulation, which leads to the instability of the distributions of basic random variables.
We apply the scheme in Section~\ref{sec:non-stable} for these two examples.

\section{Conclusion}\label{sec:conclusion} 


We present an approach to compute the moments of the distribution of random outputs in probabilistic loops with non-linear, non-polynomial updates. Our method is based on polynomial chaos expansion to approximate non-polynomial general functional assignments.
The approximations produced by our technique have optimal exponential convergence when the parameters of the general non-polynomial functions have distributions that are stable across all iterations.
We derived an upper bound on the approximation error for the case of unstable parameter distributions.
Our methods can accommodate  non-linear, non-polynomial updates in classes of probabilistic loops amenable to automated moment computation, such as the class of Prob-solvable loops.
Moreover, our techniques can be used for moment approximation for uncertainty quantification in more general probabilistic loops.
Our experiments demonstrate the ability of our methods to characterize  non-polynomial behavior  in stochastic models from various domains via their moments, with high accuracy and in a fraction of the time required by other state-of-the-art tools.

%
%
%

\newpage
\bibliographystyle{splncs04}
\bibliography{refs}
\newpage

\appendix
\section{Proof of Theorem \ref{err-bound}}\label{app:A}

\begin{proof}[Thm. \ref{err-bound}]
Since $f(z) = 0$ $\forall z\notin \left[a, b\right]$,
\begin{gather}
  \left \| g(z) - \sum\limits_{i=0}^{T}c_{i}p_{i}(z)\right \|_{f}^{2} = 
\int\limits_{a}^{b}\left(g(z) - \sum\limits_{i=0}^{T}c_{i}p_{i}(z)\right)^{2}f(z)dz \quad  \notag \\
=\int\limits_{-\infty}^{\infty}\left(g(z) - \sum\limits_{i=0}^{T}c_{i}p_{i}(z)\right)^{2}f(z)dz\notag \\
= \int\limits_{-\infty}^{a}\left(g(z) - \sum\limits_{i=0}^{T}c_{i}p_{i}(z)\right)^{2}f(z)dz 
+ \int\limits_{b}^{\infty}\left(g(z) - \sum\limits_{i=0}^{T}c_{i}p_{i}(z)\right)^{2}f(z)dz \notag \\
+ \int\limits_{a}^{b}\left(g(z) - \sum\limits_{i=0}^{T}c_{i}p_{i}(z)\right)^{2}f(z)dz  \notag \\
\leq \int\limits_{-\infty}^{a}\left(g(z) - \sum\limits_{i=0}^{T}c_{i}p_{i}(z)\right)^{2}\phi(z)dz 
 + \int\limits_{b}^{\infty}\left(g(z) - \sum\limits_{i=0}^{T}c_{i}p_{i}(z)\right)^{2}\phi(z)dz \notag \\
+ \int\limits_{a}^{b}\left(g(z) - \sum\limits_{i=0}^{T}c_{i}p_{i}(z)\right)^{2}\phi(z)dz  
+ \int_{a}^{b}\left(g(z) - \sum\limits_{i=0}^{T}c_{i}p_{i}(z)\right)^{2}(f(z)-\phi(z))dz \notag \\
= A + B + C + D \label{star}
\end{gather}
Since $ f(z) -\phi(z)  \leq  \phi(z) + f(z)$, $D$ satisfies 
\begin{gather*}
\int\limits_{a}^{b}\left(g(z) - \sum\limits_{i=0}^{T}c_{i}p_{i}(z)\right)^{2}(f(z)-\phi(z))dz \\ \leq \int\limits_{a}^{b}\left(g(z) - \sum\limits_{i=0}^{T}c_{i}p_{i}(z)\right)^{2}dz\int\limits_{a}^{b}(\phi(z) + f(z))dz  \\
= (1 + \Phi(b) - \Phi(a)) \times  \int\limits_{a}^{b}\left(g(z) - \sum\limits_{i=0}^{T}c_{i}p_{i}(z)\right)^{2}dz,
\end{gather*}
with $(1 + \Phi(b) - \Phi(a)) < 2.$
Now, \[ 1 \leq \frac{\phi(z)}{\min{(\phi(a), \phi(b))}} \quad \forall z \in \left[a, b\right],\] and hence 
\begin{gather}
   \int\limits_{a}^{b}\left(g(z) - \sum\limits_{i=0}^{T}c_{i}p_{i}(z)\right)^{2}dz \notag \\
   \leq  \min{(\phi(a), \phi(b))}^{-1}\int\limits_{a}^{b}\left(g(z) - \sum\limits_{i=0}^{T}c_{i}p_{i}(z)\right)^{2}\phi(x)dz \notag\\
\leq \min{(\phi(a), \phi(b))}^{-1} C. \label{2nd} 
\end{gather} 
By \eqref{2nd} and \eqref{truef}, \eqref{star} satisfies
\begin{gather}
 A+B+C+D \leq \left(\frac{2}{\min{(\phi(a), \phi(b))}} + 1 \right)  
\Bigg\{ \int\limits_{-\infty}^{a} \left(g(z) - \sum\limits_{i=0}^{T}c_{i}p_{i}(z)\right)^{2}\phi(z)dz \notag \\ + 
\int\limits_{b}^{\infty}\left(g(z) - \sum\limits_{i=0}^{T}c_{i}p_{i}(z)\right)^{2}\phi(z)dz 
 +  \int\limits_{a}^{b}\left(g(z) - \sum\limits_{i=0}^{T}c_{i}p_{i}(z)\right)^{2}\phi(z)dz  \Bigg\}  \notag \\
= \left(\frac{2}{\min{(\phi(a), \phi(b))}} + 1 \right) 
  \int\limits_{-\infty}^{\infty}\left(g(z) - \sum\limits_{i=0}^{T}c_{i}p_{i}(z)\right)^{2}\phi(z)dz \notag\\
   = \left(\frac{2}{\min{(\phi(a), \phi(b))}} + 1 \right)\sum\limits_{i=T+1}^{\infty}c_{i}^{2} \leq \left( \frac{2}{\min{(\phi(a), \phi(b))}} + 1 \right) \var_{\phi}\left[g(Z)\right]  \notag 
\end{gather} 
since $\var_\phi(g(Z))=\sum_{i=1}^\infty c_i^2$.
In consequence, the error \eqref{approxer} 
can be upper bounded by \eqref{bound}.
\end{proof}

\section{Computation Algorithm in Detail}\label{app:B}

We let $\D \in \mathbb{Z}^{L \times k}$ be
 the matrix with each row $j=1, \ldots, L$ containing the degrees of $Z_i$ (in column $i$) in the corresponding polynomial in \eqref{polyexpanse}. The first row corresponds to the constant polynomial (1), and the last row to $\bar{p}_{1}^{\bar{d}_1}(z_1) \ldots \bar{p}_{k}^{\bar{d}_k}(z_k)$. That is,
\[
\D = (d_{ji})_{j = 1,\ldots,L,\; i = 1,\ldots,k} = 
\begin{blockarray}{ccccc}
Z_{1} & Z_{2} & Z_{3} & ... & Z_{k} \\
\begin{block}{(ccccc)}
  0 & 0 & 0 & 0 & 0 \\
  0 & 0 & 0 & 0 & 1 \\
  \vdots & \vdots & \vdots  & \ddots & \vdots\\
  \bar{d}_{1} & \bar{d}_{2} & \bar{d}_{3} & \cdots & \bar{d}_{k} \\
\end{block}
\end{blockarray}
 \]
The computer implementation of the algorithm computes $c_j$ in \eqref{coef} for each $j$ combination (row) of degrees for the corresponding polynomials.
 
We apply the above computation to the following example.  
Suppose that  $X$ has a truncated normal distribution with parameters $\mu = 2$, $\sigma = 0.1$, and is supported over $\left[1, 3\right]$, and that $Y$ is uniformly distributed over $\left[1, 2\right]$. We expand  $g(x, y) = \log(x + y)$ along $X$ and $Y$, as follows. We choose the relevant highest degrees of expansion to be $\bar{d}_{X} = 2$ and $\bar{d}_{Y} = 2$. The pdf  of $Y$ is $f_{Y}(y) \equiv 1$, and of $X$ is $f_{X}(x) = e^{-\frac{(x-2)^{2}}{0.02}}/0.1TrMul\sqrt{2\pi}$, where the truncation multiplier, $TrMul$, equals $\int\limits_{1}^{3}e^{-\frac{(x-2)^{2}}{0.02}}dx/0.1\sqrt{2\pi}$.

The two sets of polynomials, $\{p_{i}\} = \{1, x, x^{2}\}$ and $\{q_{i}\} = \{1, y, y^{2}\}$, are linearly independent. Applying the Gram-Schmidt procedure to orthogonalize and normalize them, we obtain $\bar{p}_{0}(x) = 1$, $\bar{p}_{1}(x) = 10x - 20$, $\bar{p}_{2}(x) = 70.71067x^{2} -282.84271x  + 282.13561$, and $\bar{q}_{0}(y) = 1$, $\bar{q}_{1}(y) = 3.4641y - 5.19615$, $\bar{q}_{2}(y) =  13.41641y^{2} - 40.24922y + 29.06888$. 

In this case, $L = (1 + \bar{d}_{X}) * (1 + \bar{d}_{Y}) = 9$, and  $\D$ has 9 rows and 2 columns, 
\[
\D = (d_{ji}) =
\begin{blockarray}{ccc}
X & Y & \\
\begin{block}{(cc)c}
  0 & 0  & {} \leftarrow c_{1} = \hspace{7pt} 1.2489233\\
  0 & 1  & {} \leftarrow c_{2} = \hspace{7pt} 0.0828874\\
  0 & 2  & {} \leftarrow c_{3} = -0.0030768\\
  1 & 0  & {} \leftarrow c_{4} = \hspace{7pt} 0.0287925\\
  1 & 1  & {} \leftarrow c_{5} = -0.0023918\\
  1 & 2  & {} \leftarrow c_{6} = \hspace{7pt} 0.0001778\\
  2 & 0  & {} \leftarrow c_{7} = -0.0005907\\
  2 & 1  & {} \leftarrow c_{8} = \hspace{7pt} 0.0000981\\
  2 & 2  & {} \leftarrow c_{9} = -0.0000109\\
\end{block}
\end{blockarray}
 \]
Iterating through the rows of matrix $\D$ and choosing the relevant combination of degrees of polynomials for each variable, we  calculate the Fourier coefficients,
$$c_{j} = \int_{1}^{3}\int_{1}^{2}log(x + y)\bar{p}_{d_{j1}}(x)\bar{q}_{d_{j2}}(y)f_{X}(x)f_{Y}(y)dydx.$$
The final estimator can be derived by summing up the products of each coefficient and the relevant combination of polynomials:
$$\log(x + y) \approx \hat{g}(x, y) =  \sum_{j=1}^{9}c_{j}\bar{p}_{d_{j1}}(x)\bar{q}_{d_{j2}}(y) \approx $$
$$-0.01038x^{2}y^{2} + 0.05517x^{2}y- 0.10031x^{2} + $$ 
$$ + 0.06538xy^{2} - 0.37513xy + 0.86515x - $$
$$- 0.13042y^{2} + 0.93998y - 0.59927.$$
The estimation error is 
\begin{equation}\notag
    se(\hat{g}) = \sqrt{\int\limits_{1}^{3}\int\limits_{1}^{2}\left(\log(x+y)) - \hat{g}(x, y)\right)^{2}f_{X}(x)f_{Y}(y)dydx} \approx 0.000151895
\end{equation}

\section{PCEs of exponential and trigonometric functions}\label{app:C}

Table \ref{tab:LiOFunc} lists examples of functions of one up to three random arguments approximated by PCE's of different degrees and, correspondingly, number of coefficients. We use $TruncNormal (\mu, \sigma^{2}, \left[a, b\right])$ to denote the truncated normal distribution with expectation $\mu$ and  standard deviation $\sigma$ on the (finite or infinite) interval $\left[a, b\right]$, and $TruncGamma(\theta, k, \left[a, b\right])$ for the  truncated gamma distribution on the (finite or infinite) interval $\left[a, b\right]$, $a,b>0$, with shape parameter $k$ and scale parameter $\theta$.
The  approximation error in \eqref{error_std} is reported in the last column. The results confirm \eqref{convergence} in practice: the error decreases as the degree or, equivalently, the number of  components in the approximation of the polynomial increases.  

\begin{table}[h!]
\centering
\renewcommand{\arraystretch}{1.3}
\setlength\belowcaptionskip{8pt}
\resizebox{\textwidth}{!}{
\begin{tabular}{@{}llcc@{}} 
 \toprule
 Function & Random Variables & Degree / \#coefficients & Error\\ \midrule
 
 \begin{tabular}{@{}l@{}}$f(x_{1}, x_{2}) = \xi e^{-x_{1}} + (\xi - \frac{\xi^{2}}{2}) e^{x_{2}-x_{1}}$ \\ $\xi = 0.3$ \end{tabular} & 
  \begin{tabular}{@{}l@{}}$x_{1} \sim Normal(0, 1),$ \\ $x_{2} \sim Normal(2, 0.01)$ \end{tabular} & 
  \begin{tabular}{@{}c@{}} 1 / 4 \\ 2 / 9 \\ 3 / 16 \\ 4 / 25 \\ 5 / 36 \end{tabular} & 
  \begin{tabular}{@{}c@{}} 3.076846 \\ 1.696078 \\ 0.825399 \\ 0.363869 \\ 0.270419 \end{tabular} \\ 
  \hline

 $f(x_{1}, x_{2}) = 0.3 e^{x_{1} - x_{2}} + 0.6 e^{ - x_{2}}$ & 
  \begin{tabular}{@{}l@{}}$x_{1} \sim TruncNormal(4, 1, \left[3, 5\right]),$ \\ $x_{2} \sim TruncNormal(2, 0.01, \left[0, 4\right])$ \end{tabular} & 
  \begin{tabular}{@{}c@{}} 1 / 4 \\ 2 / 9 \\ 3 / 16 \\ 4 / 25 \\ 5 / 36 \end{tabular} & 
  \begin{tabular}{@{}c@{}} 0.343870 \\ 0.057076 \\ 0.007112 \\ 0.000709 \\ 0.000059 \end{tabular} \\ 
  \hline

 $f(x_{1}, x_{2}) = e^{x_{1}x_{2}}$ & 
 \begin{tabular}{@{}l@{}}$x_{1} \sim TruncNormal(4, 1, \left[3, 5\right])$ \\ $x_{2} \sim TruncGamma(3, 1, \left[0.5, 1\right])$ \end{tabular} & 
 \begin{tabular}{@{}c@{}} 1 / 4 \\ 2 / 9 \\ 3 / 16 \\ 4 / 25 \\ 5 / 36 \end{tabular} & 
 \begin{tabular}{@{}c@{}} 5.745048 \\ 1.035060 \\ 0.142816 \\ 0.016118 \\ 0.001543 \end{tabular}\\ 
 \hline

 \begin{tabular}{@{}l@{}}$f(x_{1}, x_{2}, x_{3}) = 0.3 e^{x_{1} - x_{2}} + $ \\ \hspace{2.25cm} $ 0.6e^{x_{2} - x_{3}} +  0.1 e^{x_{3} - x_{1}}$  \end{tabular} &
 \begin{tabular}{@{}l@{}}$x_{1} \sim TruncNormal(4, 1, \left[3, 5\right])$ \\ $x_{2} \sim TruncGamma(3, 1, \left[0.5, 1\right])$ \\ $x_{3} \sim U\left[4, 8\right]$\end{tabular} & 
 \begin{tabular}{@{}c@{}} 1 / 8 \\ 2 / 27 \\ 3 / 64 \end{tabular} & 
 \begin{tabular}{@{}c@{}} 1.637981 \\ 0.303096 \\ 0.066869 \end{tabular}\\
 \hline

 \begin{tabular}{@{}l@{}}$f(x_{1}) = \psi cos(x_{1}) + (1 - \psi)sin(x_{1})$ \\ $\psi = 0.3$  \end{tabular} & $x_{1} \sim Normal(0, 1)$ & 
 \begin{tabular}{@{}c@{}} 1 / 2 \\ 2 / 3 \\ 3 / 4 \\ 4 / 5 \\ 5 / 6 \end{tabular} &  
  \begin{tabular}{@{}c@{}} 0.222627 \\ 0.181681 \\ 0.054450 \\ 0.039815 \\ 0.009115 \end{tabular}\\ 
\bottomrule
\end{tabular}
}
\vspace{0.5em}
\caption{
Approximations of $5$ non-linear functions using PCE.
}
\label{tab:LiOFunc}
\end{table}

\end{document}